\documentclass{ieeetj}
\usepackage{cite}
\usepackage{amsmath,amssymb,amsfonts,amsthm}
\usepackage{graphicx,color}
\usepackage{booktabs}
\usepackage{array}
\usepackage{xcolor}
\usepackage[hidelinks]{hyperref}

\newtheorem{theorem}{Theorem}
\newtheorem{lemma}{Lemma}
\newtheorem{corollary}{Corollary}
\newtheorem{definition}{Definition}
\newtheorem{remark}{Remark}
\newtheorem{assumption}{Assumption}

\newcommand{\Xb}{\overline{X}}
\newcommand{\E}{\mathbb{E}}
\newcommand{\V}{\mathbb{V}}
\newcommand{\prob}{\mathbb{P}}
\newcommand{\lab}[1]{\mbox{\normalfont\scriptsize\scshape #1}}

\AtBeginDocument{\definecolor{tmlcncolor}{cmyk}{0.93,0.59,0.15,0.02}%
  \definecolor{NavyBlue}{RGB}{0,86,125}}
\def\OJlogo{\vspace{-4pt}}
\def\seclogo{\vspace{10pt}}
\def\authorrefmark#1{\ensuremath{^{\textbf{#1}}}}

\begin{document}
\receiveddate{XX Month, XXXX}
\reviseddate{XX Month, XXXX}
\accepteddate{XX Month, XXXX}
\publisheddate{XX Month, XXXX}
\currentdate{XX Month, XXXX}
\doiinfo{XXXX.XXXX.XXXXXXX}

\renewcommand{\topfraction}{0.92}
\renewcommand{\bottomfraction}{0.7}
\renewcommand{\textfraction}{0.06}
\renewcommand{\floatpagefraction}{0.85}
\renewcommand{\dbltopfraction}{0.92}
\renewcommand{\dblfloatpagefraction}{0.85}
\setcounter{topnumber}{3}
\setcounter{dbltopnumber}{3}
\setcounter{totalnumber}{5}

\markboth{}{Li \emph{et al.}}

\title{Probabilistic Performance Analysis of Parallel Signature Search
Strategies in Multi-Level Tree Networks}

\author{Jingwei Li\authorrefmark{1} and Thomas G. Robertazzi\raisebox{0pt}[0pt][0pt]{\textsuperscript{\footnotesize 1,2}}, Fellow, IEEE}
\affil{Department of Electrical and Computer Engineering, Stony Brook University, 100 Nicolls Rd, Stony Brook, NY 11794, USA}
\affil{Department of Applied Mathematics and Statistics (affiliate), Stony Brook University, 100 Nicolls Rd, Stony Brook, NY 11794, USA}
\corresp{Corresponding author: Thomas G. Robertazzi (email: thomas.robertazzi@stonybrook.edu).}
\newif\ifpreprint
\preprinttrue

\authornote{This work is based in part on the first author's doctoral dissertation, Stony Brook University, 2022. This research did not receive any specific grant from funding agencies in the public, commercial, or not-for-profit sectors.\ifpreprint\ This work has been submitted to the IEEE for possible publication. Copyright may be transferred without notice, after which this version may no longer be accessible.\fi}

\begin{abstract}
Hierarchical distributed search --- locating a data pattern, or
\emph{signature}, across a tree-structured collection of files ---
underlies distributed index traversal, deep packet inspection and
sequence alignment. A practitioner must decide how much
parallelism to employ: scan each layer sequentially, fan out within
subtrees, or launch the whole tree at once. Existing analyses answer
this only partially: they characterize every node by the statistics
of a signature-holding file and, for multi-signature files, need
quantities revealed only at run time. We develop a probabilistic
framework predicting the completion time of five search strategies,
spanning sequential to full-tree parallelism, \emph{before any file is
read}. Node scan times are modeled
as a mixture over signature presence, layer times as order statistics,
and parallel subtree scans by extreme-value arguments; when signature
counts are known, occupancy under capacity constraints is treated by
generating functions. Each performance formula carries an exactness
label ---
exact (or exact-in-regime), plug-in, asymptotic or bound --- each
approximation quantified against Monte Carlo simulation and its regime
identified. A multicore prototype reproduces the coarse separation
between full-tree, layer- and subtree-level parallelism, but shows
that synchronization overhead can erase the predicted separation
between close strategies. The framework delivers \emph{a priori}
completion-time predictions with explicit accuracy regimes and
negligible computational cost, the design example evaluated in under a
millisecond; these timing models can support subsequent resource-cost
optimization.
\end{abstract}


\maketitle

\section{Introduction}\label{sec:intro}

A \emph{signature} is a data pattern of interest inside a large file or a
collection of files: a gene fragment in a DNA database, a byte pattern in
network traffic, a keyword in a document corpus, a feature in an image
archive. When the collection is distributed over a network of
processors, the time to locate all signatures depends jointly on the
network topology, the search strategy, and the statistics of signature
occurrence. Multi-level tree networks are a natural substrate for such
searches: work flows from a root to successively wider layers, matching
both the structure of hierarchical storage systems and the master--worker
pattern of divisible load scheduling.

This paper asks a planning question: \emph{before} a search is launched,
what is its expected completion time? An a priori answer lets an
operator compare candidate strategies, dimension a tree, or negotiate a
deadline without running the search. On elastic infrastructure the
question is sharper still: when processors are rented by the slot, how
widely to fan out is a cost decision as much as a latency one, and it
must be made before the workload reveals itself. Prior analyses of signature search
time in tree networks \cite{ying2014} answer this only partially: they
model every node with the statistics of a signature-holding node, and,
when files may hold several signatures, require quantities that only
become known during the search itself.

To the best of our knowledge, this is the first unified a priori
framework that jointly covers both information conditions, all three
capacity classes, and the full range of layer- and subtree-level
parallelism, with every expression traceable to a stated exactness
label. Our contributions are as follows.

First, a framework spanning two information conditions (signature
counts per layer unknown vs.\ known), three capacity classes (at most
one, at most $K$, unlimited) and five strategies from fully sequential
to fully parallel, each with an explicit analytical or numerically
evaluable expected-time characterization --- except for one
complementary capacity-constrained regime, which we identify and leave
open. What the unification requires matters more than the count of
settings: a mixture correction to the node-time model, explicit
conditioning on random layer widths, the occupancy law of subtree
counts, and computational routes for small and large parameters alike.

Second, a sharper probabilistic foundation: node times are modeled as a
mixture over signature presence, parallel stages via order statistics,
and parallel subtree scans via central-limit and extreme-value
arguments with the correct mixture variance. Third, for known counts we
show that per-subtree occupancies follow a multivariate hypergeometric
law --- not the uniform-composition law that stars-and-bars counting
suggests --- characterize the minimum subtree count through a
generating function, and solve the multi-signature case under
\emph{both} placement conventions in one framework. Fourth, the
principal formulas are validated by Monte Carlo simulation, and every
formula is classified as exact or approximate with a quantified
accuracy range, and a multicore prototype measures how far the
unbounded-processor idealization carries in real execution.

The paper is organized as follows. Section~\ref{sec:related} reviews
related work. Section~\ref{sec:model} defines the model.
Sections~\ref{sec:unknown} and \ref{sec:known} analyze the unknown- and
known-count conditions. Section~\ref{sec:validation} reports the
simulation study and Section~\ref{sec:design} a design example with the
computational cost of the analysis;
Section~\ref{sec:discussion} discusses scope and
limitations, and Section~\ref{sec:conclusion} concludes.

\section{Related Work}\label{sec:related}

\emph{Divisible load theory.} DLT studies the optimal partitioning of
arbitrarily divisible computation among networked processors
\cite{bharadwaj1996,robertazzi2003,bharadwaj2003,robertazzi-book}.
Closed-form and algorithmic results exist for buses, daisy chains,
stars and multi-level trees \cite{beaumont2005}, with continuing
extensions, e.g., to coarse-grained workloads
\cite{song2023}. Our work belongs to the performance-analysis branch of
this literature: rather than optimizing a load partition, we predict the
completion time of a fixed hierarchical search discipline.

\emph{Signature searching.} Ying and Robertazzi \cite{ying2014}
introduced the time-cost analysis of signature searching in a networked
collection of files, establishing the scenario taxonomy and the basic
expected-time arguments on which the present work is built. Their
analysis characterizes each node by the statistics of a
signature-holding file, which is exact when signatures are plentiful but increasingly
optimistic as $p$ decreases, underestimating the latency (Section~\ref{sec:validation}),
and it treats the single- and multi-signature cases separately, the
latter in terms of quantities revealed during the search. We refine
that foundation in four directions: node times are modeled as a mixture
over signature presence; a capacity class interpolating between one and
unlimited signatures is added; every setting is made computable strictly
in advance; and two subtree-oriented disciplines are introduced,
requiring central-limit and extreme-value arguments.
Table~\ref{tab:vs-prior} summarizes the differences dimension by
dimension.

\begin{table}[!t]
\caption{Scope of \cite{ying2014} and of this work}
\label{tab:vs-prior}
\centering
\small
\begin{tabular}{p{1.5cm}p{2.3cm}p{3.1cm}}
\toprule
Dimension & Prior work \cite{ying2014} & This work \\
\midrule
Node time & signature-node law &
  signature/non-signature mixture \\
Information & partly runtime & strictly a priori \\
Capacity & one, unlimited & one, $K$, unlimited (unified) \\
Strategies & layer-level & five, incl.\ subtree-level \\
Known counts & limited &
  occupancy law under two conventions \\
Validation & limited &
  Monte Carlo, error regimes, prototype \\
Resource view & --- & peak concurrency reported \\
\bottomrule
\end{tabular}
\end{table}

\emph{Distributed pattern search.} On the application side, the
primitive analyzed here appears wherever a pattern must be located
across a distributed corpus: regular-expression and string matching for
deep packet inspection and intrusion detection \cite{xu2016}, and
sequence alignment over clustered genomic data, where read alignment is
routinely distributed with data-parallel frameworks
\cite{dean2008,zaharia2012} such as Apache Spark
\cite{abuin2016,aljame2023}. Hierarchical summaries --- Bloom filters
(space-efficient probabilistic structures for testing set membership)
and related structures \cite{broder2004}, or the merged indexes of
Section~\ref{sec:model} --- are what make the pruning assumption
realizable in such systems. That literature optimizes the matcher run
at a node; we take the per-node scan as given and predict the
completion time of the search that coordinates the nodes.

\emph{Parallel completion-time modeling.} A layer of our tree is a
fork--join stage: work is split, executed concurrently, and joined, so
the stage time is the maximum of the branch times. Exact analysis is
hard beyond two branches, and the literature relies on approximations
and bounds \cite{nelson1988}, surveyed in \cite{thomasian2014}, with
recent extensions to heterogeneous servers \cite{mohanty2024} and to
bounded, heterogeneous processor sets \cite{wang2024} --- the setting
our prototype meets empirically. The same maximum-of-many effect appears operationally as the straggler
problem \cite{dean2013}, mitigated by cause-aware restart and
placement \cite{ananth2010} or by cloning \cite{ananth2013}. Our $S_4$ analysis is its analytical
counterpart: the layer time is an extreme value over subtrees, growing
like $\sqrt{\ln M}$ rather than with the mean. Unlike queueing fork--join models, our branches are not
service stations fed by arrivals but bounded scans of known length,
which is what makes closed forms attainable.

\emph{Combinatorial occupancy.} The known-count analysis draws on
classical occupancy theory \cite{johnson1977,feller1968}, whose
balls-into-bins asymptotics are standard tools for load distribution
\cite{raab1998}. Capacity-constrained counts of placements of
identical balls into distinct urns, used in
Section~\ref{sec:known}-\ref{sec:known-K}, are derived in our companion combinatorial
paper \cite{li2025urn}; the equal-composition vs.\ uniform-placement
distinction of Section~\ref{sec:discussion} mirrors the classical
Bose--Einstein vs.\ Maxwell--Boltzmann dichotomy \cite{feller1968}.

\section{System Model and Assumptions}\label{sec:model}

\subsection{Multi-Level Tree Networks}

We consider a rooted tree network of height $H$. The root (layer $0$)
acts as coordinator; every non-root node holds one searchable object,
which we call its file: at a leaf a \emph{shard} (one partition of a
horizontally split data set), internally a \emph{summary index} over
the shards below. A node in
layer $i-1$ that participates in the search has $n_i$ children in layer
$i$, $i=1,\dots,H$; all subtrees rooted in the same layer are therefore
homogeneous in size. The $n_i$ children of a common parent form a
\emph{subtree} of layer $i$. Derivations keep $n_i$ general; a bare
$n$ means $n_i=n$ for all $i$. A table of notation is given in the supplementary material.

\subsection{Search-Time Model}\label{sec:time}

Scanning
a file takes time $s$; we normalize the scan to proceed at unit speed, so
a signature located at position $u\in[0,s]$ is found at time $u$.
Signature positions are uniformly distributed, hence for a file that
contains a signature the discovery time is $U[0,s]$ with mean
$\Xb=s/2$. If a file contains no signature, the scan can only terminate
after the entire file has been read, taking the deterministic time
$s=2\Xb$.

\begin{assumption}[Signature inheritance]\label{as:inherit}
If a node holds no signature, none of its descendants holds a signature,
and its subtree is pruned from the search. If a node holds a signature,
each of its children holds a signature with probability $p$ at
capacity one, independently across children. Writing $K$ for the
maximum number of signatures a single file may hold
(Section~\ref{sec:model}-\ref{sec:classes}), $p$ is for $K>1$ the
per-slot occupancy and a child is nonempty with probability
$q=1-(1-p)^K$ (Remark~\ref{rem:p}).
\end{assumption}

\begin{remark}[Two readings of $p$]\label{rem:p}
At capacity one, $p$ is a \emph{node-level} presence probability. For
$K>1$ (Section~\ref{sec:unknown}-\ref{sec:unknown-K}) $p$ is the \emph{slot-level}
occupancy probability, each of the $K$ slots being filled independently,
so the node-level probability of holding at least one signature is
$q=1-(1-p)^{K}$, which reduces to $p$ at $K=1$.
\end{remark}

\begin{figure}[!t]
\centering
\includegraphics[width=0.82\columnwidth]{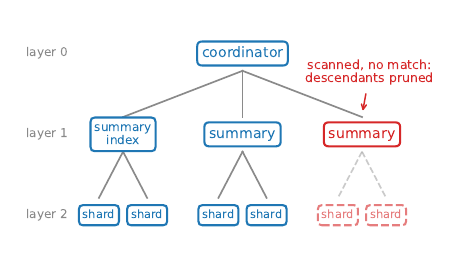}
\caption{Reference architecture, shown for $H=2$ with $n_1=3$ and
$n_2=2$. The root acts as query coordinator; internal nodes hold a
summary index over their subtree; leaves hold data shards (partitions
of the data set). A summary
is itself scanned; one ending without a match prunes its descendant
subtree (dashed).}
\label{fig:arch}
\end{figure}

Assumption~\ref{as:inherit}, from \cite{ying2014}, is not arbitrary
but the semantics of a summarized hierarchical index, the running
example of this paper (Figure~\ref{fig:arch}). Leaves store data
shards; an internal node stores a summary index over the shards below
--- the merged list of patterns occurring anywhere in its subtree, or
an equivalent digest --- which is itself scanned. Because the index is
complete, a scan that reaches its end without a match certifies that no
descendant can match and the subtree is skipped: exactly the pruning of
Assumption~\ref{as:inherit}. The same completeness fixes the two scan
times of Section~\ref{sec:model}-\ref{sec:time}. A node whose index
holds no match must be read in full before that conclusion is safe,
costing $s$; a node whose index does hold one stops at the matching
entry, at a position we take to be uniform, costing $U[0,s]$ in
expectation. Versioned storage and provenance graphs read the same way.

\subsection{Search Scenarios}

Layers are always processed one after another; the five scenarios differ
in how the nodes \emph{within} a layer are processed:

\begin{itemize}
  \item[$S_1$:] all nodes of the layer are searched sequentially;
  \item[$S_2$:] all nodes of the layer are searched in parallel;
  \item[$S_3$:] within each subtree (the $n_i$ children of one parent)
    the nodes are searched in parallel, subtrees sequentially;
  \item[$S_4$:] the nodes of each subtree are searched sequentially,
    subtrees in parallel;
  \item[$S_5$:] every node of the entire tree is searched in parallel
    (single stage, no layer synchronization).
\end{itemize}

\begin{remark}[Execution semantics of $S_5$]\label{rem:S5}
$S_1$--$S_4$ synchronize between layers, so a node is started only after
its parent is known to hold a signature and pruning saves work.
$S_5$ has no such synchronization point, and we define it as
\emph{speculative full-tree execution}: all $N=\sum_{i\le H}\prod_{j\le
i}n_j$ nodes are launched at $t=0$, and the scans of subtrees that
pruning would have eliminated are performed but discarded. Completion
time is thus the maximum over all launched nodes, and the processor
cost is charged for all $N$ of them --- the accounting used in
Section~\ref{sec:design}. $S_5$ consequently upper-bounds the
achievable parallelism of any dependency-respecting schedule and serves
as an idealized latency lower bound; its interest lies precisely in the
resource price of that bound.
\end{remark}

\subsection{Random Quantities and Conditioning}\label{sec:rv}

Because several quantities below are random, we fix their status once
and for all. \emph{Deterministic parameters}: $H$, $n_i$, $s$, $p$,
$K$, and --- in the known-count condition --- the signature counts
$m_i$. \emph{Random variables}: the per-node search time $t$; the
number $M_i$ (resp.\ $\widehat M_i$) of signature-holding nodes in layer
$i$ when files hold at most one (resp.\ at most $K$) signature; the
per-subtree occupancy vector; and every layer completion time. The
number of nodes searched in layer $i$ is $M_{i-1}n_i$ (resp.\
$\widehat M_{i-1}n_i$ in the capacity-$K$ case), itself random
for $i\ge 2$. Note that $M_i$ counts only the signature-holding nodes
among those $M_{i-1}n_i$: at capacity one $M_i\mid M_{i-1}\sim
\mathrm{Bin}(M_{i-1}n_i, p)$. \emph{Hat convention}: throughout, a hat
marks the capacity-$K$ counterpart of a capacity-one quantity ---
$\widehat M_i$, $\widehat T_j$, $\widehat\tau_i$ --- each reducing to
it at $K=1$.

Layer times are therefore analyzed \emph{conditionally} on the
\emph{width} of the preceding layer --- the number $M_{i-1}$ of its
signature-holding nodes, which fixes how many nodes layer $i$ searches
--- and the total is obtained by averaging over that width. Writing $\tau_i(\cdot)$ for the conditional expected time of
layer $i$ given the number of searched parents, the exact total is
\begin{equation}\label{eq:cond}
T=\sum_{i=1}^{H}\E\bigl[\tau_i(M_{i-1})\bigr]
 =\sum_{i=1}^{H}\sum_{k}\prob(M_{i-1}=k)\,\tau_i(k) .
\end{equation}
When $\tau_i$ is affine in $M_{i-1}$, \eqref{eq:cond} collapses to
$\tau_i(\E M_{i-1})$ and the result is exact; otherwise we evaluate
$\tau_i(\E M_{i-1})$ as a plug-in (mean-field) approximation and report
its error. Note that $\E M_{i-1}$ is generally not an integer, so
whenever $\tau_i$ involves combinatorial quantities defined only for
integer arguments, the plug-in is used with $\E M_{i-1}$ rounded to the
nearest integer; \eqref{eq:cond} remains available as the exact
alternative at higher cost.

\begin{remark}[Exactness labels]\label{rem:labels}
Every result below carries one of five labels:
\lab{exact} (rigorous in the stated model);
\lab{plug-in} (uses $\E g(M)\approx g(\E M)$);
\lab{clt/evt} (relies on asymptotic distributions);
\lab{bound} (systematic under- or over-estimate); and
\lab{exact-in-regime} (rigorous under a stated parameter condition).
Section~\ref{sec:validation} quantifies the error of every non-exact
label.
\end{remark}

\subsection{Information Conditions and Capacity Classes}\label{sec:classes}

Two information conditions are studied. In the \emph{unknown-count}
condition (Section~\ref{sec:unknown}), signature occurrence is only known
statistically through Assumption~\ref{as:inherit}. In the
\emph{known-count} condition (Section~\ref{sec:known}), the exact number
$m_i$ of signatures in each layer is known in advance, and the signatures
are placed at random over the searched nodes of that layer: uniformly
over node subsets when files hold at most one signature, and, for
capacity $K>1$, according to one of the two placement conventions
introduced in Section~\ref{sec:known}-\ref{sec:known-K}.

\begin{assumption}[No global cancellation]\label{as:nocancel}
Knowing $m_i$ tells the coordinator when a layer's signatures have all
been found, but we do not assume it can retract dispatched work: every
dispatched node completes its scan. Under this convention
$S_1$--$S_4$ differ only in how work is parallelized; $S_5$
additionally executes speculatively (Remark~\ref{rem:S5}), so its work
set is the full tree.
The supplementary material quantifies what cancellation would save for
$S_1$ and $S_2$, both in closed form, and why subtree-level
cancellation ($S_3$, $S_4$) does not yield a comparable closed form
under the present analysis.
\end{assumption}

Orthogonally, three capacity classes are considered: each file holds at
most one signature; at most $K$ signatures; or an unlimited number of
signatures. The results of this paper are \emph{a priori} computable:
they depend only on $(H, n_i, s)$ together with $p$ (and $K$), or with
$\{m_i\}$ --- quantities available before the search starts.

\section{Unknown Number of Signatures}\label{sec:unknown}

\subsection{At Most One Signature per File}\label{sec:unknown-one}

\subsubsection{Single-Node Statistics}

Under Assumption~\ref{as:inherit} a searched node holds a signature with
probability $p$. Its search time $t$ therefore follows the mixture
\begin{equation}\label{eq:mixture}
f(t) \;=\; p\,f_1(t) + (1-p)\,f_2(t),
\end{equation}
where $f_1$ is the density of $U[0,s]$ and $f_2$ is a unit point mass at
$t=s$.

\begin{lemma}\label{lem:node}
With $\Xb=s/2$, the search time of a searched node satisfies
\begin{align}
\E(t) &= p\,\frac{s}{2} + (1-p)\,s = (2-p)\Xb, \label{eq:Et}\\
\V(t) &= \E(t^2)-\E(t)^2
       = \frac{s^2 p\,(4-3p)}{12}
       = \frac{p\,(4-3p)}{3}\,\Xb^{2}. \label{eq:Vt}
\end{align}
\end{lemma}

\begin{proof}
The mixture mean is the probability-weighted mean,
giving \eqref{eq:Et}. For the second moment,
$\E(t^2)=p\!\int_0^s \frac{u^2}{s}\,du+(1-p)s^2
       =\bigl(\tfrac{p}{3}+1-p\bigr)s^2$,
and \eqref{eq:Vt} follows from
$\V(t)=\E(t^2)-\E(t)^2$ with $\E(t)=(1-\tfrac{p}{2})s$.
\end{proof}

\begin{remark}
Note that the variance of a mixture is \emph{not} the weighted sum of the
component variances: the dispersion between the component means
contributes the additional term visible in \eqref{eq:Vt}. We write
$\sigma \triangleq \sqrt{\V(t)} = \Xb\sqrt{p(4-3p)/3}$.
\end{remark}

\subsubsection{Expected Number of Searched Nodes}

Let $M_i$ denote the number of signature-holding nodes in layer $i$
($M_0=1$). Each of the $M_{i-1}n_i$ searched nodes of layer $i$ holds a
signature independently with probability $p$, hence
$\E(M_i \mid M_{i-1}) = p\,n_i\,M_{i-1}$ and, by induction,
\begin{equation}\label{eq:EMi}
\E(M_i) = p^{\,i}\prod_{j=1}^{i} n_j , \qquad i=1,\dots,H .
\end{equation}

\begin{remark}[Plug-in approximation]\label{rem:jensen}
$M_i$ enters several of the formulas below nonlinearly. Wherever this
happens we use the plug-in approximation
$\E\,g(M_{i-1}) \approx g\bigl(\E M_{i-1}\bigr)$,
and we quantify the resulting error by exact enumeration of the width
law in Section~\ref{sec:validation}. Formulas that are linear in $M_{i-1}$ ($T_1$, $T_3$) are exact. We
set $\tau_i(0)=0$: a layer whose parent width is zero is never
executed, and $\mu(1)=0$ in \eqref{eq:mu}, whose expansion is
undefined at $M=1$. The plug-in also presumes the width distribution
is concentrated. In subcritical trees ($pn_i\lesssim 1$) extinction
carries substantial probability, and evaluating a nonlinear $\tau_i$ at
a mean width that the process may never reach is a structural, not a
rounding, error: at $H=3$, $n=6$, $p=0.1$ the exact average
\eqref{eq:cond} gives $T_2=1.72$ against $3.00$ for the plug-in. We
recommend, as a conservative rule of thumb, \eqref{eq:cond} when
$pn_i$ is around $2$ or below.
\end{remark}

\subsubsection{Parallel Search of $n$ Nodes}

\begin{lemma}\label{lem:parallel}
Let $\hat t$ be the completion time of $n$ nodes searched in parallel,
each holding a signature independently with probability $p$. Then
\begin{equation}\label{eq:parallel}
\E(\hat t) = \Bigl(1-\frac{p^{\,n}}{n+1}\Bigr)\,2\Xb .
\end{equation}
\end{lemma}

\begin{proof}
If at least one node holds no signature (probability $1-p^n$), that node
scans its whole file, so $\hat t = s$ exactly. If all $n$ nodes hold
signatures (probability $p^n$), $\hat t=\max(t_1,\dots,t_n)$ with
$t_k \sim U[0,s]$ i.i.d.; the maximum has CDF $(y/s)^n$ and mean
$\tfrac{n}{n+1}s$. Combining,
$\E(\hat t)=p^n \tfrac{n}{n+1}s+(1-p^n)s$, which is \eqref{eq:parallel}.
\end{proof}

\subsubsection{Total Expected Search Times}

\begin{theorem}\label{th:unknown-one}
Under the model of Section~\ref{sec:model} with at most one signature per
file, the total expected search times are
\begin{align}
T_1 &= \sum_{i=1}^{H} p^{\,i-1}\Bigl(\prod_{j=1}^{i} n_j\Bigr)(2-p)\Xb,
  &&\lab{exact}\label{eq:T1}\\
T_2 &\approx \sum_{i=1}^{H}
  \Bigl(1-\frac{p^{\,\nu_i}}{\nu_i+1}\Bigr)2\Xb,
  \;\; \nu_i = \E(M_{i-1})n_i,
  &&\lab{plug-in}\label{eq:T2}\\
T_3 &= \sum_{i=1}^{H} \E(M_{i-1})
  \Bigl(1-\frac{p^{\,n_i}}{n_i+1}\Bigr)2\Xb,
  &&\lab{exact}\label{eq:T3}\\
T_4 &\approx \sum_{i=1}^{H}
  \Bigl[\sqrt{n_i}\,\sigma\,\mu\bigl(\E(M_{i-1})\bigr)
        + n_i(2-p)\Xb\Bigr],
  &&\lab{\!\!plug-in, clt/evt}\label{eq:T4}\\
T_5 &= \Bigl(1-\frac{p^{\,N}}{N+1}\Bigr)2\Xb,
  \;\; N=\sum_{i=1}^{H}\prod_{j=1}^{i} n_j,
  &&\lab{exact}\label{eq:T5}
\end{align}
where $\sigma=\Xb\sqrt{p(4-3p)/3}$ and
\begin{equation}\label{eq:mu}
\mu(M) = \sqrt{2\ln M}
       - \frac{\ln\ln M + \ln 4\pi - 2\gamma}{2\sqrt{2\ln M}},
\end{equation}
with $\gamma\approx 0.5772$ the Euler--Mascheroni constant.
\end{theorem}

\begin{proof}[Proof sketch]
$T_1$: conditionally on $M_{i-1}$ the layer-$i$ workload is
$M_{i-1}n_i$ sequentially searched nodes of mean $(2-p)\Xb$ each
(Lemma~\ref{lem:node}); the conditional expectation is affine in
$M_{i-1}$, so \eqref{eq:cond} collapses and \eqref{eq:EMi} gives
\eqref{eq:T1} exactly. $T_2$ and $T_3$ follow from
Lemma~\ref{lem:parallel} applied to the whole layer and to one subtree:
$T_3$ is again affine in $M_{i-1}$ and hence exact, whereas $T_2$'s
conditional expectation is nonlinear in $M_{i-1}$ and is evaluated by
plug-in (Remark~\ref{rem:jensen}). $T_5$ involves no random width ---
by Remark~\ref{rem:S5} all $N$ nodes are launched --- so
Lemma~\ref{lem:parallel} applies directly with $n=N$ and the result is
exact.

For $T_4$, the layer-$i$ time is the maximum over $M_{i-1}$ parallel
subtrees of the sum $U_j=\sum_{k=1}^{n_i}t_{j,k}$ of $n_i$ i.i.d.\
node times. By the central limit theorem
$U_j \approx n_i(2-p)\Xb + \sqrt{n_i}\,\sigma Z_j$ with
$Z_j\sim N(0,1)$; the expected maximum of $M$ i.i.d.\ standard normal
variables admits the expansion $\mu(M)$ of \eqref{eq:mu}
\cite{david2003}, in the tradition of \cite{cramer1946}. Combining the two yields \eqref{eq:T4}.
\end{proof}

\begin{remark}[Applicability of $T_4$]\label{rem:T4}
\eqref{eq:T4} rests on two asymptotic steps: the CLT (accurate for
$n_i\gtrsim 30$) and the extreme-value expansion (accurate for
$M_{i-1}\gtrsim 10$). Within it the formula stays within about $1\%$ of simulation
(supplementary accuracy table). Outside this
regime, the numerical convolution route in \eqref{eq:t4num} should be
used instead.
\end{remark}

\subsubsection{Exact Numerical Evaluation of the $S_4$ Layer}
\label{sec:t4num}

The asymptotic step in \eqref{eq:T4} can be avoided entirely. The
subtree time $U=\sum_{k=1}^{n}t_k$ is a sum of i.i.d.\ node times whose
law is known exactly --- density $p/s$ on $[0,s)$ plus an atom of mass
$1-p$ at $s$ --- so its CDF $F_U$ is the $n$-fold convolution of that
law, computable on a grid, and the layer time follows from the
identity
\begin{equation}\label{eq:t4num}
\E\bigl[\max(U_1,\dots,U_M)\bigr]
 = ns-\int_0^{ns}F_U(x)^{M}\,dx .
\end{equation}
This is \lab{exact} up to discretization and costs $O(n^2 G)$ for a
grid of $G$ points, against $O(1)$ for \eqref{eq:T4}. We therefore
select automatically: use \eqref{eq:T4} when $n_i\ge 50$ and
$M_{i-1}\ge 10$, and \eqref{eq:t4num} otherwise. Over configurations spanning $n\in[4,100]$ and $M\in[3,50]$
(supplementary material), the numerical route stays within $0.06\%$ of
simulation everywhere, while the asymptotic formula degrades from
$0.5\%$ at $n=100$, $M=50$ to $3.3\%$ at $n=4$, $M=5$ --- so the
selection rule buys accuracy exactly where it is needed.

\subsection{General Discovery-Time Distributions}\label{sec:general}

Signature positions need not be uniform. Let the discovery time in a
signature-holding file have an arbitrary CDF $F_{\mathrm s}$ on
$[0,s]$, with mean $\mu_{\mathrm s}$ and second moment
$\mu^{(2)}_{\mathrm s}$. The node time is still the mixture
\eqref{eq:mixture}, so
\begin{equation}\label{eq:gen-node}
\E(t)=p\,\mu_{\mathrm s}+(1-p)s ,\qquad
\V(t)=p\,\mu^{(2)}_{\mathrm s}+(1-p)s^{2}-\E(t)^{2},
\end{equation}
and the node CDF is $F_t(x)=pF_{\mathrm s}(x)$ for $x<s$, with the
remaining mass at $s$. Sequential stages ($T_1$, and the subtree sums inside $T_4$ via the
CLT) depend on $F_{\mathrm s}$ only through these moments; parallel
stages do not. For $n$ nodes searched concurrently,
\begin{equation}\label{eq:gen-max}
\E\bigl[\max(t_1,\dots,t_n)\bigr]
 = s-\int_0^{s}\bigl[p\,F_{\mathrm s}(x)\bigr]^{n}\,dx ,
\end{equation}
which requires the entire CDF. Uniform positions give $\int_0^s (px/s)^n dx = p^n s/(n+1)$ and
recover Lemma~\ref{lem:parallel}; the $K$-signature case follows by replacing the early-finish weight
$p$ with $p^{K}$ and $F_{\mathrm s}$ with the CDF
$F_{\mathrm s}(x)^{K}$ of the $K$th order statistic.
Every result below thus transports to non-uniform positions via
\eqref{eq:gen-node}--\eqref{eq:gen-max}, at the price of one numerical
integral where no closed form exists.

\subsection{At Most $K$ Signatures per File}\label{sec:unknown-K}

We now allow each file to hold up to $K$ signatures. The number of
signatures in a searched node is modeled as $\mathrm{Binomial}(K,p)$:
each of the $K$ capacity slots is occupied independently with probability
$p$. A node therefore holds at least one signature with probability
$q \triangleq 1-(1-p)^K$, and, writing $\widehat M_i$ for the number of
signature-holding nodes in layer $i$, the argument leading to
\eqref{eq:EMi} gives
\begin{equation}\label{eq:EMhat}
\E(\widehat M_i)=q^{\,i}\prod_{j=1}^{i}n_j ,\qquad \widehat M_0=1 .
\end{equation}

The scan of a node can terminate early only when the $K$th signature is
found, since only then can the searcher be certain that nothing remains.
Consequently, with probability $p^K$ (exactly $K$ signatures) the node
time is the largest of $K$ i.i.d.\ $U[0,s]$ positions, with density
$K t^{K-1}/s^K$ and mean $\tfrac{K}{K+1}s$; otherwise the node time is
$s$. This yields the following analogue of Lemma~\ref{lem:node}
(derivations in the supplementary material).

\begin{lemma}\label{lem:nodeK}
With at most $K$ signatures per file,
\begin{align}
\E(t) &= \Bigl(1-\frac{p^{K}}{K+1}\Bigr)2\Xb , \label{eq:EtK}\\
\V(t) &= \Bigl[p^{K}\tfrac{K}{K+2} + 1-p^{K}
        - \bigl(1-\tfrac{p^{K}}{K+1}\bigr)^{2}\Bigr]s^{2}
        \triangleq \sigma_K^2 . \label{eq:VtK}
\end{align}
For $K=1$ these reduce to \eqref{eq:Et} and \eqref{eq:Vt}.
\end{lemma}

\begin{lemma}\label{lem:parallelK}
If $\nu$ such nodes are searched in parallel, the expected completion
time is
\begin{equation}\label{eq:parallelK}
\E(\hat t)=\Bigl(1-\frac{p^{\,\nu K}}{\nu K+1}\Bigr)2\Xb .
\end{equation}
\end{lemma}

\begin{proof}
The completion time is smaller than $s$ only if \emph{every} node holds
exactly $K$ signatures, an event of probability $(p^{K})^{\nu}=p^{\nu K}$,
in which case it is the maximum of $\nu$ i.i.d.\ variables with CDF
$(t/s)^K$, i.e., has CDF $(t/s)^{\nu K}$ and mean
$\tfrac{\nu K}{\nu K+1}s$. Hence
$\E(\hat t)=p^{\nu K}\tfrac{\nu K}{\nu K+1}s+(1-p^{\nu K})s$, which
simplifies to \eqref{eq:parallelK}. For $K=1$ this is
Lemma~\ref{lem:parallel}.
\end{proof}

\begin{theorem}\label{th:unknown-K}
With at most $K$ signatures per file, writing
$\hat\nu_i=\E(\widehat M_{i-1})\,n_i$ and
$N=\sum_{i=1}^{H}\prod_{j=1}^{i}n_j$,
\begin{align}
\widehat T_1 &= \sum_{i=1}^{H}\E(\widehat M_{i-1})\,n_i
   \Bigl(1-\frac{p^{K}}{K+1}\Bigr)2\Xb , \label{eq:T1K}\\
\widehat T_2 &\approx \sum_{i=1}^{H}
   \Bigl(1-\frac{p^{\,\hat\nu_i K}}{\hat\nu_i K+1}\Bigr)2\Xb ,
   \label{eq:T2K}\\
\widehat T_3 &= \sum_{i=1}^{H}\E(\widehat M_{i-1})
   \Bigl(1-\frac{p^{\,n_i K}}{n_i K+1}\Bigr)2\Xb , \label{eq:T3K}\\
\widehat T_4 &\approx \sum_{i=1}^{H}
   \Bigl[\sqrt{n_i}\,\sigma_K\,\mu\bigl(\E(\widehat M_{i-1})\bigr)
   + n_i\Bigl(1-\frac{p^{K}}{K+1}\Bigr)2\Xb\Bigr] , \label{eq:T4K}\\
\widehat T_5 &= \Bigl(1-\frac{p^{\,N K}}{N K+1}\Bigr)2\Xb .
   \label{eq:T5K}
\end{align}
$\widehat T_1$ and $\widehat T_3$ are \lab{exact} (their conditional
expectations are affine in $\widehat M_{i-1}$, so \eqref{eq:cond}
collapses); $\widehat T_5$ is \lab{exact} by Remark~\ref{rem:S5};
$\widehat T_2$ is \lab{plug-in}; $\widehat T_4$ is
\lab{plug-in, clt/evt}, inheriting Remark~\ref{rem:T4}.
\end{theorem}

\subsection{Unlimited Signatures per File}\label{sec:unknown-inf}

Removing the capacity bound corresponds to $K\to\infty$. For
$0<p<1$, $p^{K}\to 0$ and $q\to 1$, so early termination becomes
impossible ($\E(t)\to 2\Xb$, $\sigma_K\to 0$) and every node is
searched ($\E(\widehat M_i)\to\prod_{j\le i}n_j$); the endpoints
$p\in\{0,1\}$ are degenerate and excluded. Theorem~\ref{th:unknown-K} then collapses accordingly.

Writing $N_i=\prod_{j\le i}n_j$, the limits are
$T_1=2\Xb\sum_i N_i$, $T_2=2H\Xb$,
$T_3=2\Xb\sum_i N_{i-1}$, $T_4=2\Xb\sum_i n_i$ and $T_5=2\Xb$.

\section{Known Number of Signatures}\label{sec:known}

We now assume the exact number of signatures per layer,
$m_1,\dots,m_H$ (with $m_0=1$), is known in advance --- e.g., from an
index, a checksum manifest, or a prior coarse scan. For the capacity-one case, layer $i$ contains
$m_{i-1}$ searched subtrees ($n_i$ nodes each), and the $m_i$ signatures
occupy a uniformly random $m_i$-subset of the $m_{i-1}n_i$ searched
nodes.

\subsection{At Most One Signature per File}\label{sec:known-one}

\subsubsection{Sequential Search ($S_1$)}

Write $\eta_i = m_{i-1}n_i$ for the number of searched nodes of layer
$i$. Under Assumption~\ref{as:nocancel} every searched node is scanned:
$m_i$ of them stop at a signature (mean $\Xb$), the other
$\eta_i-m_i$ scan in full. By linearity,
\begin{equation}\label{eq:knownT1}
T_1^{\mathrm{kn}}
 = \sum_{i=1}^{H}\bigl(2\eta_i-m_i\bigr)\Xb ,
\end{equation}
which is \lab{exact} and requires no independence between nodes.

\subsubsection{Fully Parallel Layers ($S_2$)}

If $m_i<\eta_i$ at least one searched node is empty and the layer takes
exactly $2\Xb$ (an empty node always scans its full file, and no node
can exceed $s$). If $m_i=\eta_i$, the layer time is the maximum of
$\eta_i$ uniforms. With the indicator $\mathbf{1}\{\cdot\}$,
\begin{equation}\label{eq:knownT2}
T_2^{\mathrm{kn}}
 = \sum_{i=1}^{H}
   \Bigl(1-\frac{\mathbf{1}\{m_i=\eta_i\}}{m_i+1}\Bigr)2\Xb
 \qquad\text{(exact)} .
\end{equation}

\subsubsection{Parallel Subtrees, Sequential Across ($S_3$)}

Let $F_i=\prob(\text{a given subtree is fully occupied})
=\binom{(m_{i-1}-1)n_i}{\,m_i-n_i\,}\big/\binom{\eta_i}{m_i}$
(zero when $m_i<n_i$), i.e., the hypergeometric probability that a fixed
subtree receives $n_i$ of the $m_i$ signatures. A subtree that is not
fully occupied contains an empty node and takes exactly $2\Xb$; a fully
occupied one takes $\tfrac{n_i}{n_i+1}2\Xb$ in expectation. By linearity
of expectation over the $m_{i-1}$ subtrees,
\begin{equation}\label{eq:knownT3}
T_3^{\mathrm{kn}}
 = \sum_{i=1}^{H} m_{i-1}\,
   \Bigl(1-\frac{F_i}{n_i+1}\Bigr)2\Xb
 \qquad\text{(exact)} .
\end{equation}

\subsubsection{Sequential Subtrees, Parallel Across ($S_4$)}

Let $X_1,\dots,X_{m_{i-1}}$ be the signature counts of the subtrees of
layer $i$. Since the $m_i$ signatures occupy a uniform subset of the
$\eta_i$ nodes, $(X_1,\dots,X_{m_{i-1}})$ follows a \emph{multivariate
hypergeometric} law --- not the uniform-composition (``stars and bars'')
law, under which all count vectors would be equally likely. The two
models differ substantially; see Section~\ref{sec:validation}.

\begin{lemma}\label{lem:xmin}
Let $X_{\min}=\min_k X_k$ and
$\varphi_j(z)=\sum_{x=j}^{n_i}\binom{n_i}{x}z^{x}$. Then
\begin{equation}\label{eq:xmin}
\prob(X_{\min}\ge j)
 = \frac{[z^{m_i}]\;\varphi_j(z)^{\,m_{i-1}}}{\binom{\eta_i}{m_i}} ,
\end{equation}
where $[z^m]$ extracts the coefficient of $z^m$;
$\prob(X_{\min}=j)=\prob(X_{\min}\ge j)-\prob(X_{\min}\ge j+1)$.
The coefficient is computable by an $O(m_{i-1}m_i n_i)$ convolution.
\end{lemma}

\begin{proof}
A placement with all counts $\ge j$ chooses $x_k\ge j$ nodes inside
subtree $k$ with $\sum_k x_k=m_i$; the number of such placements is the
coefficient of $z^{m_i}$ in $\prod_k \varphi_j(z)$, and all
$\binom{\eta_i}{m_i}$ placements are equally likely.
\end{proof}

A subtree with $x$ signatures takes $x\Xb+(n_i-x)2\Xb=(2n_i-x)\Xb$ in
expectation. The subtree with the fewest signatures therefore has the
largest \emph{conditional mean} time, though not necessarily the largest
realized time; identifying the layer maximum with it gives
\begin{equation}\label{eq:knownT4}
T_4^{\mathrm{kn}}
 \approx \sum_{i=1}^{H}\;
 \sum_{j=0}^{\lfloor m_i/m_{i-1}\rfloor}
 \prob(X_{\min}=j)\,\bigl(2n_i-j\bigr)\Xb .
\end{equation}

\begin{remark}[Nature of the approximation]\label{rem:A4}
\eqref{eq:knownT4} identifies the layer time with the expected time of
the minimum-count subtree. Conditionally on the counts, however, the
layer time is the maximum of random sums, which can be attained by
another subtree; hence \eqref{eq:knownT4} is a lower bound, observed
within about $2\%$ of simulation in the configurations of
Section~\ref{sec:validation}.
\end{remark}

\subsubsection{Fully Parallel Tree ($S_5$)}

Under Remark~\ref{rem:S5}, $S_5$ launches all $N_i=\prod_{j\le i}n_j$
nodes of layer $i$ regardless of pruning, so it finishes before $s$
only if $m_i=N_i$ for every $i$:
\begin{equation}\label{eq:knownT5}
T_5^{\mathrm{kn}}
 = \Bigl(1-\frac{\prod_{i=1}^{H}\mathbf{1}\{m_i=N_i\}}
   {1+\sum_i m_i}\Bigr)2\Xb
 \qquad\text{(exact)} .
\end{equation}

\subsection{At Most $K$ Signatures per File}\label{sec:known-K}

When a file can hold up to $K$ signatures, the $m_i$ signatures of layer
$i$ are placed into the $\widehat M_{i-1}n_i$ searched nodes subject to the
capacity bound. Two placement conventions are natural, and we analyze
both.

\begin{assumption}[Equal-composition placement, BE]\label{as:BE}
All capacity-feasible count vectors (compositions) are equally likely
--- the Bose--Einstein convention for indistinguishable signatures.
\end{assumption}

\begin{assumption}[Uniform placement, MB]\label{as:MB}
Each signature independently selects a node uniformly at random,
conditioned on no node exceeding $K$ --- the Maxwell--Boltzmann
convention, natural when signatures arise at nodes by independent
physical processes.
\end{assumption}

For $K=1$ the two conventions coincide with each other and with the
uniform-subset model of Section~\ref{sec:known}-\ref{sec:known-one}; for $K>1$ they
genuinely differ (Section~\ref{sec:discussion}). Both are captured by
one generating-function device.

\begin{lemma}[Capacity-constrained counts]\label{lem:counts}
Let $\phi_K(z)=\sum_{r=0}^{K}z^{r}$ and
$\psi_K(z)=\sum_{r=0}^{K}z^{r}/r!$. The number of placements of $m$
signatures into $n$ nodes of capacity $K$ is
\[
\Omega^{\,n,K}_{m}=[z^{m}]\phi_K(z)^{n}
\;\text{(BE)},\quad
A^{\,n,K}_{m}=m!\,[z^{m}]\psi_K(z)^{n}
\;\text{(MB)},
\]
where $\Omega^{\,n,K}_{m}$ is the capacity-constrained count derived in
our companion paper \cite{li2025urn}; every
capacity-feasible configuration is equally likely within its convention. Both are computable by an $O(nmK)$ convolution.
\end{lemma}

Three node-level probabilities drive the analysis; writing
$\eta$ for the (conditioned) number of searched nodes in the layer and
suppressing the layer index, they are, under
BE and MB respectively:
\begin{align}
P_0 &= \frac{\Omega^{\,\eta-1,K}_{m}}{\Omega^{\,\eta,K}_{m}}
     \quad\text{or}\quad
     \frac{A^{\,\eta-1,K}_{m}}{A^{\,\eta,K}_{m}}
     && \text{(node empty)}, \label{eq:P0}\\
P_K &= \frac{\Omega^{\,\eta-1,K}_{m-K}}{\Omega^{\,\eta,K}_{m}}
     \quad\text{or}\quad
     \binom{m}{K}\frac{A^{\,\eta-1,K}_{m-K}}{A^{\,\eta,K}_{m}}
     && \text{(node exactly $K$)}, \label{eq:PK}\\
F &= \frac{\Omega^{\,\eta-n_i,K}_{m-Kn_i}}{\Omega^{\,\eta,K}_{m}}
     \;\text{or}\;
     c\,\frac{A^{\,\eta-n_i,K}_{m-Kn_i}}{A^{\,\eta,K}_{m}}
     && \text{(saturated)}, \label{eq:F}
\end{align}
where $c=\binom{m}{Kn_i}(Kn_i)!/(K!)^{n_i}$ counts the ways of choosing
and arranging the signatures inside a saturated subtree. Each formula
counts the placements of the remaining signatures in the remaining
capacity. A node is nonempty with probability $1-P_0$, which gives the
\lab{plug-in} recursion for the expected number of searched nodes,
\begin{equation}\label{eq:mhat}
\E\widehat M_i\;\approx\;\E\widehat M_{i-1}\, n_i\,
 \bigl(1-P_0(m_i,\E\widehat M_{i-1}n_i,K)\bigr),
\qquad \widehat M_0=1 ,
\end{equation}
Equation \eqref{eq:mhat} is a mean-field recursion: $P_0$ is nonlinear
in the width, so the exact propagation carries the full law of
$\widehat M_i$ (Remark~\ref{rem:integrality}) rather than its mean.
Below, $P_{i,K}$ denotes \eqref{eq:PK} evaluated at layer $i$.

\begin{remark}[Integrality and the exact alternative]\label{rem:integrality}
Unlike the capacity-one case, where the number of searched subtrees is
the known integer $m_{i-1}$, the width $\widehat M_{i-1}$ is genuinely
random for $K>1$: a node may hold several signatures, so the count
does not determine the number of occupied nodes. Hence \eqref{eq:mhat} propagates an
\emph{expected} width, generally not an integer, whereas
$\Omega^{\,\eta,K}_{m}$ and $A^{\,\eta,K}_{m}$ are defined only for
integer $\eta$ (we evaluate at $\lfloor\eta_i\rceil$); and
$P_{i,K}$, $F_i$ are nonlinear in $\eta_i$, so the plug-in
approximates even before rounding. The exact alternative is \eqref{eq:cond}: propagate the full law of
$\widehat M_{i}$, whose conditional transition
$\prob(\widehat M_i=k\mid \widehat M_{i-1}=r)$ the same generating
functions supply (supplementary material), the marginal following by
summing over $r$. This needs generating-function computations per
layer; Section~\ref{sec:validation} reports a $0.07\%$ plug-in error
in the case examined, so we keep it by default, using the exact route
when $\widehat M_{i-1}$ is small.
\end{remark}

Since $\widehat M_{i-1}$ is random (Remark~\ref{rem:integrality}), we
state the results \emph{conditionally} on its realization and average
afterwards, following \eqref{eq:cond}.

\begin{theorem}[Conditional layer times]\label{th:known-K}
Fix a layer $i$ and condition on $\widehat M_{i-1}=r$, an integer, so
that the layer has $\eta_i(r)=r\,n_i$ searched nodes. Let
$P_{i,K}(r)$ and $F_i(r)$ denote \eqref{eq:PK} and \eqref{eq:F}
evaluated at $\eta_i(r)$ for the chosen convention
(Assumption~\ref{as:BE} or \ref{as:MB}). Then the conditional expected
time of layer $i$ under strategy $S_j$ is
\begin{align}
\tau_{1,i}(r) &= r\,n_i
  \Bigl(1-\frac{P_{i,K}(r)}{K+1}\Bigr)2\Xb , \label{eq:knT1K}\\
\tau_{2,i}(r) &= \Bigl(1-
  \frac{\mathbf{1}\{m_i=\eta_i(r) K\}}{\eta_i(r) K+1}\Bigr)2\Xb ,
  \label{eq:knT2K}\\
\tau_{3,i}(r) &= r
  \Bigl(1-\frac{F_i(r)}{Kn_i+1}\Bigr)2\Xb , \label{eq:knT3K}\\
\tau_{4,i}(r) &= 2 n_i \Xb
  \qquad\qquad\quad (m_i<r\,K), \label{eq:knT4K}\\
\tau_{5}\phantom{_{,i}}\;\; &=
  \Bigl(1-\frac{\prod_{i}\mathbf{1}\{m_i=N_i K\}}
       {1+K\sum_{i} N_i}\Bigr)2\Xb ,
  \;\; N_i=\prod_{j\le i}n_j .
  \label{eq:knT5K}
\end{align}
Given the realized width $r$, \eqref{eq:knT1K}--\eqref{eq:knT3K} are
\lab{exact}: \eqref{eq:knT1K} and \eqref{eq:knT3K} because their
conditional expectations are affine in the number of nodes and of
subtrees, \eqref{eq:knT2K} because the saturated and non-saturated
cases are exhaustive --- and \eqref{eq:knT4K} is
\lab{exact-in-regime}, valid whenever the \emph{realized} width
satisfies $m_i<rK$. Expression \eqref{eq:knT5K} carries no layer index
and no width argument: under the speculative semantics of
Remark~\ref{rem:S5} all $N_i=\prod_{j\le i}n_j$ nodes of layer $i$ are
launched irrespective of occupancy, so $S_5$ is a single stage over the
whole tree and \eqref{eq:knT5K} is \lab{exact}.
\end{theorem}

\begin{corollary}[Totals and their plug-in]\label{cor:uncond}
For $j\in\{1,2,3,4\}$ the total expected search time is the average of
the conditional layer times over the width distribution,
\begin{equation}\label{eq:uncond}
\widehat T^{\mathrm{kn}}_j
 =\sum_{i=1}^{H}\E\bigl[\tau_{j,i}(\widehat M_{i-1})\bigr]
 =\sum_{i=1}^{H}\sum_{r}\prob\bigl(\widehat M_{i-1}=r\bigr)\,
   \tau_{j,i}(r),
\end{equation}
with $\prob(\widehat M_{i}=r)$ available from
Remark~\ref{rem:integrality}, and $\widehat T^{\mathrm{kn}}_5=\tau_5$.
For $j\in\{1,2,3\}$ the averaging \eqref{eq:uncond} is \lab{exact};
for $j=4$ it is \lab{exact-in-regime}, being exact only when the
scarcity condition $m_i<rK$ of \eqref{eq:knT4K} holds for every $r$ of
positive probability, and otherwise leaving the residual mass to the
open case of Remark~\ref{rem:regime}. Evaluating each layer at the
rounded expected width instead of averaging gives the \lab{plug-in}
form
\begin{equation}\label{eq:plugin}
\widehat T^{\mathrm{kn,plug}}_j
 =\sum_{i=1}^{H}\tau_{j,i}\bigl(\lfloor\E\widehat M_{i-1}\rceil\bigr),
\end{equation}
which is what the numerical results of Section~\ref{sec:validation}
implement; it is cheaper than \eqref{eq:uncond} by a factor equal to
the support size of the width distribution.
\end{corollary}

\begin{remark}[The regime condition must be checked per realization]
\label{rem:regime}
The scarcity condition of \eqref{eq:knT4K} concerns the realized width:
$m_i<\E[\widehat M_{i-1}]\,K$ does \emph{not} imply $m_i<rK$ for every
$r$ in the support. The exact average \eqref{eq:uncond} checks the condition at every $r$
of positive probability; the plug-in checks only the rounded mean
width, so it may apply \eqref{eq:knT4K} to realizations violating it.
The residual mass needs the capacity-constrained minimum-count
distribution and is left open. In the configurations reported here the condition holds
with probability above $0.99$, but a deployment near the boundary
should evaluate \eqref{eq:uncond} term by term.
\end{remark}

\begin{proof}[Proof sketch (full derivations in the supplementary material)]
\eqref{eq:knT1K}: per-node expectation as in Lemma~\ref{lem:nodeK}
with $p^K$ replaced by $P_{i,K}(r)$, summed over the $\eta_i(r)$ nodes.
\eqref{eq:knT2K} and \eqref{eq:knT5K}: a node finishes before $s$ only
with exactly $K$ signatures; for \emph{all} nodes this forces
$m_i=\eta_i(r)K$, which the capacity constraint admits in one
configuration under either convention, and the time is then the largest
of $\eta_i(r)K$ uniforms. \eqref{eq:knT3K}: a subtree finishes early
only if saturated (probability $F_i(r)$), giving the maximum of $Kn_i$
uniforms; summing over subtrees is affine in the width.
\eqref{eq:knT4K}: when $m_i<rK$ the signatures cannot supply an
exactly-$K$ node to every subtree, so some subtree scans all $n_i$
files fully and the layer time is the deterministic $2n_i\Xb$
(Remark~\ref{rem:regime}).
\end{proof}

\subsection{Unlimited Signatures per File}\label{sec:known-inf}

As $K\to\infty$ the capacity bound disappears and early termination
becomes impossible ($P_{i,K}\to 0$): every searched node scans in full.
The nonempty probability becomes $m/(m+\eta-1)$ under BE and
$1-(1-1/\eta)^{m}$ under MB, \eqref{eq:mhat} applies unchanged, and
Theorem~\ref{th:known-K} reduces to
$\widehat T^{\mathrm{kn}}_1=\sum_i \E[\widehat M_{i-1}]n_i\,2\Xb$,
$\widehat T^{\mathrm{kn}}_2=2H\Xb$,
$\widehat T^{\mathrm{kn}}_3=\sum_i\E[\widehat M_{i-1}]\,2\Xb$,
$\widehat T^{\mathrm{kn}}_4=\sum_i 2n_i\Xb$, and
$\widehat T^{\mathrm{kn}}_5=2\Xb$.

\section{Numerical Results and Monte Carlo Validation}\label{sec:validation}

All closed forms of Sections~\ref{sec:unknown} and \ref{sec:known} were
checked against a discrete-event Monte Carlo simulator that implements
the model of Section~\ref{sec:model} directly: trees are grown layer by
layer under Assumption~\ref{as:inherit}, or, for known counts, under
the applicable placement law (uniform node subsets at capacity one; the
BE or MB convention of Section~\ref{sec:known}-\ref{sec:known-K} for $K>1$); per-node
scan times are drawn from the mixture, and each scenario's completion
time is measured on the sampled tree. Unless noted otherwise we use $10^5$--$10^6$ replications; the
simulator and scripts are available from the authors.

\subsection{Exact Formulas}

Exactness is a property of the derivation, not of a measurement:
simulation can refute a claimed identity but cannot establish one. The
\lab{exact} labels are therefore earned in
Sections~\ref{sec:unknown}--\ref{sec:known}, and simulation serves as
an \emph{implementation check} that analysis and simulator describe the
same process. Under that reading, the
\lab{exact} expressions --- namely $T_1$, $T_3$, $T_5$, the
\lab{exact} members $\widehat T_1$, $\widehat T_3$, $\widehat T_5$ of
Theorem~\ref{th:unknown-K}, and, conditionally on
the realized width, $T^{\mathrm{kn}}_1$, $T^{\mathrm{kn}}_2$,
$T^{\mathrm{kn}}_3$, $T^{\mathrm{kn}}_5$ --- reproduce simulation to
within Monte Carlo noise ($<0.1\%$) in every configuration tested
(e.g., $T^{\mathrm{kn}}_3$: $2.6875$ predicted vs.\ $2.6870$ simulated
at $m_{i-1}=3$, $n_i=3$, $m_i=7$), and \eqref{eq:knT4K} does so on the
realizations satisfying its regime condition. Discrepancies elsewhere
are attributable to the plug-in and asymptotic steps, quantified next.

\subsection{Cost of the Plug-in Step}\label{sec:plugin}

Enumerating the width law exactly and comparing with \eqref{eq:cond}
gives the unknown-count plug-in error directly. For $T_2$ it is
$0.07\%$ in the design example ($pn=5$), $1.1\%$ at the prototype
setting ($pn=3$, where $T_4$ gives $1.9\%$), and $74\%$ in the
subcritical case of Remark~\ref{rem:jensen} ($pn=0.6$). The product
$pn_i$ is the dominant predictor across our sweep but not the only
one, since $\mathrm{Bin}(n_i,p)$ fixes the spread and the extinction
probability and not merely the mean: at fixed $pn_i=1.8$ the error
runs from $1.5\%$ at $(n,p)=(2,0.9)$ to $13.3\%$ at $(18,0.1)$. Read
$pn_i\gtrsim 2$ as a heuristic, with \eqref{eq:cond} available
whenever the margin matters.

The known-count numerics for $K>1$ use the plug-in form
\eqref{eq:plugin} of Corollary~\ref{cor:uncond} --- $\tau_{j,i}$ at
the rounded expected width --- rather than the exact average
\eqref{eq:uncond}. Enumerating the full law of $\widehat M_{i-1}$ in a
configuration small enough to admit it ($K=2$, $n_1=3$, $m_1=4$,
$n_2=2$, $m_2=3$, MB convention) gives an expected width of $2.5$ ---
non-integer, as anticipated --- and a second-layer term $\tau_{1,2}$
of $\widehat T^{\mathrm{kn}}_1$ equal to $4.1083$ under exact
conditioning against $4.1111$ under the plug-in, a $0.07\%$
discrepancy. The gap grows as the width distribution spreads, so
Remark~\ref{rem:integrality} is preferable when $\widehat M_{i-1}$ is
small.

\subsection{The Mixture Variance and the $S_4$ Regime}

Modeling both node types matters most when signatures are sparse:
assigning every node the signature-holding mean $\Xb$ understates the
sequential layer time by $(1-p)/(2-p)$, $33.3\%$ at $p=0.5$ and
approaching $50\%$ as $p\to 0$, whereas the mixture mean $(2-p)\Xb$ of
Lemma~\ref{lem:node} matches simulation to within $0.1\%$ across
$p\in[0.1,1]$.
Figure~\ref{fig:validation}(a) reports the relative error of the $S_4$ layer
formula against simulation over $(n,M)$. With the standard deviation
$\sigma$ of \eqref{eq:Vt} the error decays below $1\%$ once
$n\gtrsim 50$, consistently across $M$, and is one-sided: the
expansion overstates the layer maximum at small $n$, where $\mu(M)$
has not yet settled. Outside the asymptotic regime (small $n$ or $M\lesssim 10$) we
recommend numerical evaluation of the layer maximum, per
Remark~\ref{rem:T4}.

\begin{figure*}[!t]
\centering
\includegraphics[width=\textwidth]{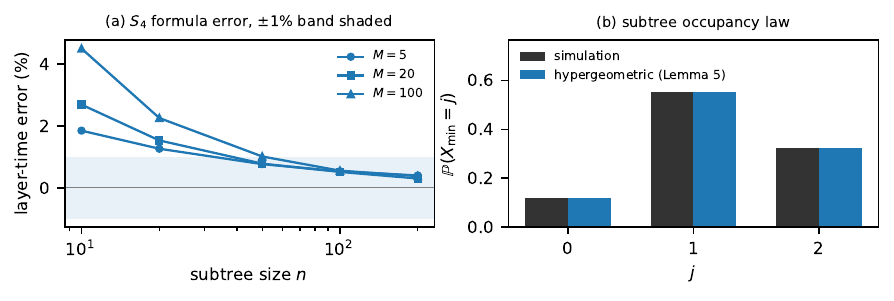}
\caption{(a) Relative error of the \emph{asymptotic} $S_4$ layer-time
formula \eqref{eq:T4} against simulation ($p=0.5$), using the variance
\eqref{eq:Vt}; the shaded band is $\pm1\%$. (b) $\prob(X_{\min}=j)$ under the multivariate
hypergeometric law of Lemma~\ref{lem:xmin} against simulation of
uniform placement \emph{within one layer}: $m_{i-1}=4$ subtrees of
$n_i=8$ nodes, carrying $m_i=10$ signatures between them.}
\label{fig:validation}
\end{figure*}

\subsection{Known-Count Occupancy Laws}\label{sec:occupancy}

Figure~\ref{fig:validation}(b) compares $\prob(X_{\min}=j)$ under the
multivariate hypergeometric law of Lemma~\ref{lem:xmin} and under the
uniform-composition law, against simulation of uniform signature
placement within a single layer of $m_{i-1}=4$ subtrees of $n_i=8$
nodes sharing $m_i=10$ signatures. The hypergeometric
distribution matches simulation pointwise, and the resulting
layer-time estimate \eqref{eq:knownT4} lies within $2\%$ of simulation
(Remark~\ref{rem:A4}). The uniform-composition law that ``stars and
bars'' counting would suggest is not equivalent: it overstates
$\prob(X_{\min}=0)$ by a factor of about five here and the layer time
by $4\%$, a gap that can widen with $m_i$.
The convention-dependent probabilities \eqref{eq:P0}--\eqref{eq:F} were
verified by rejection-sampling the conditioned multinomial law, which
realizes uniform placement directly. At $(\eta,K,m)=(4,2,5)$: $P_0$
formula $0.1500$ against $0.1510$ sampled, $P_K$ $0.4000$ against
$0.4012$, saturation probability $F$ for a two-node subtree $0.1000$
against $0.1008$, and the $S_1$ layer time $3.4667$ against $3.4661$;
agreement is equally close at $(6,3,10)$. The same two configurations
give $P_K=0.4375$ and $0.2839$ under BE --- $9\%$ and $31\%$ away from
their MB counterparts --- confirming that the choice of convention is
quantitatively consequential. The two are not interchangeable
defaults: MB suits signatures arising by independent processes, BE a
feasible allocation drawn uniformly, which connects to the counting
theory of \cite{li2025urn}.

\subsection{Scenario Comparison}

The strategies separate by growth order in the fan-out $n$: the
speedup over $S_1$ is $\Theta(n^{H})$ for $S_2$ and $S_5$, $\Theta(n)$
for $S_3$, and $\Theta(n^{H-1})$ for $S_4$ up to the logarithmic
extreme-value correction --- one power of $n$ short of full
parallelism, the price of scanning each subtree sequentially. Only the
$S_3$/$S_4$ pair has an order that changes with $(H,n)$, since theirs
are the two exponents that can coincide; Section~\ref{sec:design}
locates the crossing.

\section{Applying the Analysis}\label{sec:design}

The closed forms are meant to be used before any file is read. We
illustrate that on a design example, then measure how far the
idealization carries in real execution.

\subsection{A Design Example}\label{sec:design-ex}

An archive of roughly $10^4$ files forms a $4$-layer tree with fan-out
$n=10$; a searched file holds a signature with probability $p=0.5$.
Which parallelism is worth its cost? Table~\ref{tab:design} answers
from Section~\ref{sec:unknown}, using \eqref{eq:t4num} for $S_4$.

\begin{definition}[Reservation cost]\label{def:cost}
Let $P_{\mathrm{peak}}$ be the \emph{expected} peak concurrency --- the
largest expected number of nodes run at once, $\max_i \E(M_{i-1})n_i$
for $S_2$ --- and $T$ the completion time. We report
$C_{\mathrm{res}} = P_{\mathrm{peak}}T$, the cost of reserving that
capacity for the run. It is a \emph{planning} figure at the mean
width, not a worst-case guarantee, and it is not processor-seconds,
which equal $\sum_v \E(t_v)$ --- identical for $S_1$--$S_4$, larger
for $S_5$.

$C_{\mathrm{res}}$ prices reserved capacity but not elapsed time, so
$S_1$ minimizes it by construction: its single processor is busy
throughout, and its reserved rectangle coincides with the work
performed, which no parallel schedule can undercut. A monetary account
would add a term for the time itself --- a deadline penalty, or the
opportunity cost of the reserved window --- and the ordering it
induces can differ from, or invert, the ordering by latency, since
the two extremes differ by a factor of $1170$ in time and $11110$ in
processors. This paper supplies the physical quantities $T$ and
$P_{\mathrm{peak}}$ and their exactness; pricing them into a joint
objective, where the optimum depends on the price ratio and on the
tree shape at once, is a harder problem and is the subject of a
companion paper.
\end{definition}

\begin{table}[!t]
\caption{A priori strategy selection: $H=4$, $n=10$, $p=0.5$.
Times in units of the full-file scan $s$;
$C_{\mathrm{res}}=P_{\mathrm{peak}}T$ per Definition~\ref{def:cost},
computed from unrounded predictions. $S_4$ uses the hybrid rule of
Remark~\ref{rem:T4}, which selects the numerical route
\eqref{eq:t4num} at $n=10$.}
\label{tab:design}
\centering
\begin{tabular}{lrrrr}
\toprule
Strategy & Pred. & Sim. & $P_{\mathrm{peak}}$ & $C_{\mathrm{res}}$ \\
\midrule
$S_1$ sequential        & 1170.0 & 1173.8 & 1      & 1170 \\
$S_2$ layer-parallel    & 4.0    & 4.0    & 1250   & 5000 \\
$S_3$ subtree-parallel  & 156.0  & 156.5  & 10     & 1560 \\
$S_4$ subtrees-parallel & 35.2   & 35.1   & 125    & 4394 \\
$S_5$ all-parallel      & 1.0    & 1.0    & 11110  & 11110 \\
\bottomrule
\end{tabular}
\end{table}

The table brackets the strategies and then separates the one pair that
inspection cannot. $S_1$ and $S_5$ are the extremes: $1170$ time units
on one processor against $1.0$ on $11110$, the latter being the floor
of Remark~\ref{rem:S5}. Their order in \emph{time} is evident, but
fixing them is not idle, because they are the endpoints any cost
objective must interpolate between, and an objective that also prices
processors need not order them the same way --- $S_1$ is the cheapest
schedule in reserved capacity and the most expensive in elapsed time.
$S_2$ shows how quickly the return on parallelism decays: it reaches
$4.0$ using $1250$ processors, so the remaining $3.0$ units down to
the floor cost the other $9860$ processors.

$S_3$ and $S_4$ are the pair whose order is not fixed, and the reason
is that their speedups grow differently in $n$: $\Theta(n)$ for
$S_3$, whose layer time scales with the width $M_{i-1}$ because
subtrees run one after another, against $\Theta(n^{H-1})$ for $S_4$,
whose layer time depends on the width only through the logarithmic
$\mu$. Depth therefore helps $S_4$ and not $S_3$: at $p=0.5$ and
uniform fan-out, $S_3$ returns $4.51$ at $n=6$ for any $H$, while
$S_4$ gains a factor of $n$ per layer. The two cross near $n=8$ for
$H=3$ and below $n=4$ for $H=4$, so the configuration here puts $S_4$ ahead at
$33.2$ against $7.5$ while the shallower tree of
Table~\ref{tab:proto} reverses it, $4.5$ against $3.8$; $S_3$'s
invariance in $H$ is itself an artifact of the uniform fan-out, since
its speedup is in general a width-weighted average of per-layer
ratios. Deciding this pair is what the $S_4$ derivation of
Section~\ref{sec:unknown} is for.
Simulation confirms every prediction to within $0.5\%$, including
$S_4$, for which the hybrid rule selects the numerical route
\eqref{eq:t4num} at this fan-out.

\subsection{A Multicore Prototype}\label{sec:prototype}

The analysis idealizes execution: processors are unbounded and dispatch
is free. To see how far that carries, we ran the five strategies as
real concurrent programs on an Apple M1
(4 performance and 4 efficiency cores, macOS 26.2, Python 3.11.3,
\texttt{multiprocessing} with \texttt{fork}). Each strategy is a
dispatch schedule executed by operating-system processes over real
files, so the measured times include process creation, scheduling and
synchronization costs that the model does not predict.
Each node owns a $1$\,MB file scanned in 64\,KB chunks, with early
termination at the chunk holding the signature, so the work is
proportional to the bytes examined; the supplementary material gives
the implementation details. We report medians
and interquartile ranges over $8$ independent trees grown with a
common fan-out ($H=3$, $n_i=n=6$, $p=0.5$), each repeated twice, with
the $S_1$ baseline measured within the same tree. Layer widths vary across realizations; the tasks and
rounds in Table~\ref{tab:proto}, and the $63.2$\,MB scanned per tree
quoted below, are medians, the median tree having widths $1,3,9$ and
thus $78$ searched nodes.

What the prototype does and does not test. All five strategies run on
the \emph{same} tree realization, so their total work is identical by
construction (median $63.2$\,MB). This isolates execution order, the
worker cap and synchronization, but it also means the prototype's
$S_5$ is not the speculative variant of Remark~\ref{rem:S5}. With that
and the synthetic payloads, the measurements speak to scheduling
behavior alone: not to the extra work speculation would incur, nor to
what a real corpus would cost.

\begin{table}[!t]
\caption{Multicore prototype: measured speedup over $S_1$ (median
[IQR] over 8 trees) against the analytical prediction. ``Tasks'' is
the number of units dispatched per run; ``rounds'' is the number of
synchronization barriers.}
\label{tab:proto}
\centering
\footnotesize
\setlength{\tabcolsep}{3pt}
\begin{tabular}{@{}lrlrlrrr@{}}
\toprule
& \multicolumn{2}{c}{$W=8$} & \multicolumn{2}{c}{$W=4$} &
  & & theory \\
\cmidrule(lr){2-3}\cmidrule(lr){4-5}
Strategy & med. & [IQR] & med. & [IQR] & tasks & rounds & unbdd. \\
\midrule
$S_2$ & 3.69 & [3.32, 4.02] & 3.07 & [3.01, 3.17] & 78 & \phantom{0}3 & 19.5 \\
$S_3$ & 3.03 & [2.61, 3.23] & 2.59 & [2.53, 2.70] & 78 & 13 & \phantom{0}4.5 \\
$S_4$ & 3.02 & [2.42, 3.33] & 2.58 & [2.36, 2.73] & 13 & \phantom{0}3 & \phantom{0}3.8 \\
$S_5$ & 4.30 & [3.58, 4.66] & 3.45 & [3.34, 3.57] & 78 & \phantom{0}1 & 58.5 \\
\bottomrule
\end{tabular}
\end{table}

Three observations follow from Table~\ref{tab:proto}. First, absolute
speedups collapse towards the worker count: all strategies do identical
total work, so a work-conserving execution on $W$ processors cannot
beat $W$: for reference, the unbounded $19.5\times$ and $58.5\times$
for $S_2$ and $S_5$ are measured as $3.7\times$ and $4.3\times$ on
eight workers. Doubling
$W$ from four to eight adds only $17$--$25\%$, the additional
concurrency having to engage the efficiency cores: on a heterogeneous
machine the core count overstates the parallelism available. At $W=8$
the binding limit also differs by strategy: $8$ for $S_2$ and $S_5$,
but $4.5$ and $3.8$ for $S_3$ and $S_4$, capped by their own
structure.

Second, the coarse separation is reproduced --- $S_5$, then $S_2$,
then the subtree-level pair --- as predicted at that granularity.
Third --- and not predicted by the model --- measurement does not
separate $S_3$ and $S_4$ at all. This tree sits just below the
crossing of Section~\ref{sec:design}-\ref{sec:design-ex}, so the
formulas already place them close, $S_3$ ahead by $20\%$ ($4.5$ vs.\
$3.8$); measurement closes even that gap, to within $0.01$ at both
worker counts, with heavily overlapping interquartile ranges and a
median difference that varies between repetitions of the same
configuration ($0.08$ vs.\ $0.01$ at $W=4$; supplementary material).
This is consistent with the task
column: $S_3$ issues one dispatch round per
subtree ($13$ barriers over $78$ tasks) whereas $S_4$ issues one task
per subtree and synchronizes once per layer, so when per-node work is
small relative to dispatch cost the barrier count can offset the
advantage the latency model assigns; we did not isolate barrier cost
directly. Strategies this close in predicted latency should be
separated by their synchronization structure, not by the formulas.

A bounded-pool analysis ($\lceil\nu/W\rceil$ rounds per layer, with the
extreme-value argument inside a round) is the natural next step.

\subsection{Cost of the Analysis}

Both information conditions are cheap to evaluate, though not equally
so: $O(H)$ arithmetic operations for unknown counts, independent of
tree size, against polynomial-convolution computations per layer for
known counts, which grow with the signature counts and are the binding cost
when $m_i$ is large (supplementary material). Every quantity in
Table~\ref{tab:design} was computed in under a millisecond, against
minutes for the corresponding Monte Carlo estimate.

The margin matters less for planning than for revision. A scheduler
that re-evaluates as each layer completes knows the realized width
$M_{i-1}$ and can use the conditional $\tau_i(M_{i-1})$ directly,
which removes the width plug-in and its error. Every other
qualification stands: $\tau_{4}$ still carries its CLT/EVT
approximation or the grid error of \eqref{eq:t4num}.

\section{Discussion}\label{sec:discussion}

\subsection{Approximations and Their Ranges}

Expressions not listed below are exact or exact-in-regime under the
stated model. Four approximations are used, all quantified in
Section~\ref{sec:validation}. \emph{A1}, the plug-in
$\E g(M)\approx g(\E M)$ of Remark~\ref{rem:jensen}, affects $T_2$,
$T_4$ and their $\le K$ analogues, where $pn_i$ predicts its size well
enough to guide use: under $2\%$ on the trees reported here but $74\%$
at $pn_i=0.6$, with \eqref{eq:cond} available instead. In the
known-count totals of Corollary~\ref{cor:uncond} the same step draws
on the spread of the capacity-constrained width and is far smaller
($0.07\%$). \emph{A2} and \emph{A3}, the extreme-value
expansion $\mu(M)$ and the CLT for subtree sums, enter the $S_4$
formulas only jointly; the combined approximation stays within about
$1\%$ for $M\gtrsim 10$ and $n_i\gtrsim 50$, and \eqref{eq:t4num}
replaces both outside that regime. \emph{A4}, identifying the layer
maximum with the minimum-count subtree in \eqref{eq:knownT4}, is a
lower bound within about $2\%$.

\subsection{Limitations and Extensions}

The model assumes homogeneous subtrees, uniform signature positions,
the inheritance property of Assumption~\ref{as:inherit}, no
cancellation (Assumption~\ref{as:nocancel}), and unbounded processors.
The last is the most consequential: with $W$ workers the attainable
speedup is capped at $W$ and eroded by dispatch overhead, which the
prototype shows can leave two strategies indistinguishable that the
formulas separate by $20\%$
(Section~\ref{sec:design}-\ref{sec:prototype}). One regime is left
open: on widths $r$ with $m_i\ge rK$ the known-count $S_4$ layer time
needs a capacity-constrained minimum-count distribution we do not
derive. Evaluating \eqref{eq:uncond} term by term handles the
regime-satisfying widths exactly and isolates the rest. Heterogeneous
fan-outs
$n_{i,j}$ leave $T_1$/$T_3$-type results intact (linearity) but
complicate the extreme-value analysis; non-uniform position
distributions are handled by the general-CDF forms of
Section~\ref{sec:unknown}-\ref{sec:general}, under which sequential
stages depend on the discovery-time law only through its first two
moments while parallel stages depend on the full CDF. Beyond time, the \emph{monetary} cost
of reserving processors for tree search --- and the design of tree
shapes that minimize a joint time--cost objective --- builds directly on
the expressions derived here and is treated in a companion paper.

\section{Conclusion}\label{sec:conclusion}

We derived a priori computable expected search times for signature
searching in multi-level tree networks under five search strategies, two
information conditions (known and unknown signature counts), and three
per-file capacity classes. The analysis rests on a mixture model of
node scan times, order statistics for parallel stages, and, for known
counts, the multivariate hypergeometric law of subtree occupancy.
Formulas were checked against Monte Carlo simulation and classified as
exact, exact-in-regime, plug-in, asymptotic, or bounded, with
quantified accuracy ranges for the latter three.

The framework thus delivers a prediction available before any file is
read, together with a statement of how far to trust it: the design
example of Section~\ref{sec:design} is answered in under a millisecond
to within $0.5\%$ of simulation. Extending the analysis to a bounded
processor pool and closing the open capacity-constrained regime are the
natural next steps; the expressions derived here also supply the
temporal building blocks for cost-optimal tree design, treated in a
companion paper.


\section*{Acknowledgment}

The authors used Anthropic Claude to draft and edit the English prose
and grammar. All results and conclusions are the authors' own, who
verified every derivation, number and reference and take full
responsibility.


\end{document}